\documentclass[12pt,a4paper]{amsart}
 \usepackage{amssymb,amsthm}
\usepackage{amsmath}
\usepackage{hyperref}

\def\lrp#1{\left( #1\right)}

\def\Xn#1{#1_{1},\ldots,#1_{n}}
\def\Xnp#1{(\Xn#1)}
\def\Tref#1{Theorem~\ref{#1}}

\begin{document}

\theoremstyle{definition} 
\newtheorem{definition}{Definition}

\theoremstyle{plain}
\newtheorem{proposition}{Proposition}
\newtheorem{theorem}{Theorem}
\newtheorem{lemma}{Lemma}\title{Quantization Domains}

\author[Ratnarajan Hoole]{M.R.R. Hoole}
\address{Department of Mathematics and Statistics, University of Jaffna, Sri Lanka}

\email{kirupahoole@gmail.com}

\author[Arthur Jaffe]{Arthur Jaffe}

\address{Harvard University\\
Cambridge, MA 02138}

\email{arthur\_jaffe@harvard.edu}

\author[Christian J\"akel]{Christian D.\ J\"akel}

\address{School of Mathematics\\ 
Cardiff University, Wales}
 
\email{christian.jaekel@mac.com}

\begin{abstract}
We study the quantization of certain classical field theories using reflection positivity.  We give elementary conditions that ensure the resulting vacuum state is cyclic for products of quantum field operators, localized in a bounded Euclidean space-time region ${\mathcal O}$ at positive time.  We call such a domain a {\em quantization domain} for the classical field. The fact that bounded regions are quantization domains in classical field theory is similar to the ``Reeh-Schlieder''  property in axiomatic quantum field theory.   
\end{abstract}

\maketitle

\section{Quantization Domains}
The Reeh-Schlieder property of quantum field theory states that one can recover all  the properties described by the field, or  by bounded functions of the field called observables,  from information localized in any open (bounded) space-time domain ${\mathcal O}$.   The Reeh-Schlieder theorem states that any  Wightman field theory, or any Haag-Kastler local quantum theory, has this property~\cite{Reeh-Schlieder, Streater-Wightman, Haag}.  This result is a consequence of analyticity properties that follow in turn from the positivity of the energy and Lorentz covariance, as well as locality.  

Here we consider the analog of this result in the context of a classical field on a Euclidean space-time $X$.  We assume that the classical field can be quantized by using the Osterwalder-Schrader  property of  {reflection positivity}.   We say that a domain ${\mathcal O}$ for the classical field is a  {\em quantization domain}, if the quantization of fields supported  in ${\mathcal O}$ gives a dense set of vectors in the physical Hilbert space. 

Of course, the are well-known equivalence theorems for Wightman theory and Osterwalder-Schrader (OS) theory ensure that any bounded, positive-time domains in an OS theory is a quantization domain.  However, we are interested in  investigating  field theories for which all the standard axioms may not apply or cases in which they have not been verified.  Thus we pose weaker assumptions about the classical theory than the full OS axioms.  For example,  we replace the relativistic spectrum condition with a weaker assumption.

\subsection{Quantization}
Consider $d=s+1$ dimensional space-times 
\[
	X=\mathbb{R} \times \Sigma \; , 
\]
where $\Sigma=X_{1}\times\cdots\times X_{s}$ and $X_{j}=\mathbb{R}$ or $X_{j}=S^{1}$. Call the 
variable $t \equiv x_{0} \in \mathbb{R}$  the time coordinate.  

Now consider the Fock space $\mathcal{E}(X)$ over the one-particle space~$L_{2}(X)$. 
 Denote the Fock vacuum vector by $\Omega_{0}^{\tt E}\in \mathcal{E}(X)$ and assume that 
the classical scalar field 
\[
	\Phi(f)=\int \Phi(x)f(x)dx \; , \qquad f\in C^{\infty}_{0}(X) \; , 
\]
defined in terms of the usual creation and annihilation operators, gives rise to the characteristic functional
\[
		S(f)
		 =\langle \Omega_{0}^{\tt E}, e^{i\Phi(f)} \Omega_{0}^{\tt E}\rangle
		= \sum_{n=0}^{\infty} \frac{i^{n}} {n!}\, \langle\Omega_{0}^{\tt E}, \Phi(f)^{n} \Omega_{0}^{\tt E} \rangle \;,
\]
which is convergent for $f\in C^{\infty}_{0} (X)$.   
Assume that  the abelian space-time translation group $T(a)$ and the time reflection $\vartheta$  act covariantly on the field 
and leave the characteristic functional $S(f)$ invariant. 
   
Divide $X$ into a union of three disjoint parts 
\[
	X=X_{-}\cup X_{0} \cup X_{+} \; , 
\]
with  the reflection $\vartheta$ leaving~$X_{0}$ invariant and interchanging $X_{\pm}$.   Let $\mathcal{E}_{\pm,0}\subset \mathcal{E}$
denote the subspace of finite linear combinations 
\[
	A=\sum_{j=1}^N c_j \,e^{i\Phi(f_j)}\,\Omega_0^{\tt E} \; ,  \qquad f_j\in C^{\infty}_{0} (X_{\pm})\; , 
\]
and let $\mathcal{E}_\pm$ denote its closure in $\mathcal{E}$. Now define equivalence classes 
\begin{equation}
		\label{equiv-class}
		\widehat A=\{A+N\} \; , 
\end{equation}
with $N$ in the null space of the reflection-positive form  
\[
		\langle \widehat A, \widehat B \rangle_{\mathcal {H}}
		=\langle  A , \Theta B \rangle_{ \mathcal{E} } \qquad \text{on} \quad \mathcal{E}_{+}\times \mathcal{E}_{+} \; . 
\]
The range of the map \eqref{equiv-class} defines the pre-Hilbert space $\mathcal{H}$, whose closure gives the quantization map 
$A\mapsto \widehat A\in\mathcal{H}$.

\subsection{\label{Paragraph:SH}Standard Hypotheses {\bf C1--C3} on the Classical Fields}
Denote the  Hamiltonian by $H$, the momentum by $\vec P$, and the time-zero field (averaged 
with a real test function $h$ depending on the spatial variable) by $\varphi(0, h)=\widehat \Phi(0, h)$.  
Let $\| h \|_{\alpha}$ denote some Schwarz-space norm of the function $h$ depending on the spatial 
variables, and set 
\begin{equation}
		\| f \|_{\alpha,1} 
		= \int \| f(t,\ \cdot\ ) \|_{\alpha}\,dt \;.
	\label{Test Function Norm for Operator}
\end{equation}
Assume that 
	\begin{enumerate}
	\item[\bf C1] $S(f)$ is space-time translation and time-reflection invariant, 
\begin{equation*}
	S(f_{x})=S(\vartheta f)=S(f) \; .
\end{equation*}
	\item[\bf C2] $S(f)$ is reflection-positive on $X_{+}$.
	\item[\bf C3] There is a constant $M$, an integer $n$,  and a norm \eqref{Test Function Norm for Operator} such that 
\begin{equation*}
			| S(f) | 
			\le e^{M \| f \|_{\alpha,1}^{n}} \; . 
\end{equation*}
	\end{enumerate}
We also assume that  time translation in the positive-time direction maps $C^{\infty}_{0} (X_{\pm})$ into itself. Then the unitary time translation quantizes to a self-adjoint contraction semigroup $e^{-tH}$ generated by the Hamiltonian $H$.    Furthermore, we assume spatial translation in the $j^{\rm th}$ coordinate direction is also unitary and leaves $\Omega_{0}^{\tt E}$ invariant, so it quantizes to a unitary group $U(x_{j})=e^{-ix_{j}P_{j}}$ on $\mathcal{H}$.  Let $\vec P$ denote the vector with components $P_{j}$.  Furthermore, the classical field $\Phi$ quantizes to an imaginary time field $\varphi$. 

\subsection{\label{Paragraph:SH2}Standard Hypotheses on Quantum Mechanics of Fields}
We assume a weak form of the spectral condition:  there is a constant $M<\infty$ such that 
\begin{equation}
		0\le H\;,
		\quad
		\text{and}
		\quad
		\pm | \vec P |  \le M(H+I)\;.
	\label{Fundamental Estimates-1}
\end{equation}
Furthermore, we assume that the field at time zero can be bounded by the energy.  In particular, there is a Schwartz-space norm 
$\| h \|_{\alpha}$ on the space of time-zero  test functions  such that the  field operators satisfy  the form estimates
\begin{equation}
		\pm \varphi(0,h) 
		\le \| h \|_{\alpha} \,(H+I)\;.
	\label{Fundamental Estimates-2}
\end{equation}

\paragraph{\em Field Operators}

Let $\mathcal{D}=\cup_{0<\epsilon}e^{-\epsilon H^{2}}\mathcal{H}$, and note that the vacuum state $\Omega\in\mathcal{D}$. 
The real-time field $\varphi(x)$ and the imaginary-time field $\varphi_{\rm I}(x)$ are defined as  
sesqui-linear forms on $\mathcal{D}\times\mathcal{D}$ by
\[
		 \varphi( x)
		= e^{itH-i\vec x\cdot \vec P}\varphi(0)e^{-itH+i\vec x\cdot \vec P}\;,
\]
		and
\[
		\varphi_{\rm I}(x)
		= e^{-tH-i\vec x\cdot \vec P}\varphi(0)e^{tH+i\vec x\cdot \vec P}\;.
\]
Thus $\varphi_{\rm I}(x) = \varphi(it,\vec x)$.

\begin{proposition} [\bf Local Field]  Assume the standard hypotheses of~\S\ref{Paragraph:SH}, and let    $f\in C^{\infty}_{0}$ be real.  Then: 
	\begin{enumerate}
	\item{} The sesqui-linear form  
	$
		\varphi(f)
		= \int \varphi(x)f(x)dx
	$
uniquely defines a symmetric operator $\varphi(f)$ with domain $\mathcal{D}$, and $\varphi(f)\mathcal{D} \subset \mathcal{D}$.  
	\item{} The operator $\varphi(f)$ extends by continuity to $\mathcal{D}(H)$, and there is a constant $M_{1}<\infty$ such that
\[
		\label{M1}
		\| \varphi(f)(H+I)^{-1}\|_{\mathcal{H} }
		\le M_{1}\| f \|_{\alpha,1}\;.
\]
	\item{}  The closure $\varphi(f)^{-}$ of the operator $\varphi(f)$ is  self-adjoint. 
	\item{}   Suppose $f$ and $g$ are real valued functions in $C^{\infty}_{0}$
	and 
\[
	[ \varphi(f), \varphi(g)]\mathcal{D}=0 \; . 
\]
Then the bounded functions of $\varphi(f)^{-}$ commute with the bounded functions of $\varphi(g)^{-}$. 
	\end{enumerate} 
\end{proposition}

\subsection{Regularization}
Let the $x_{j}$'s denote the coordinates of the vector $x\in\mathbb{R}^{d}$ with $d=s+1$, and write 
\[
		D^{k} = \frac{\partial^{k_{0}+ k_{1}+\cdots+k_{s}}}{\partial x_{0}^{k_{0}}\partial x_{1}^{k_{1}} \cdots \partial x_{s}^{k_{s}}}\;,
		\text{and}  
		D^{k}_{\vec x} = \frac{\partial^{k_{1}+\cdots+k_{s}}}{\partial x_{1}^{k_{1}}\cdots \partial x_{s}^{k_{s}}}
		\; \text{if $k_0=0$}.
\]
Let $|k|=k_{0}+k_{1}+\cdots+ k_{s}$ and $k!=k_{0}! k_{1}!\cdots k_{s}! \, $.   Likewise, set 
$({\rm Ad})_{A}(B)=[A,B]$ and
\[
	({\rm Ad}_{P})^{k}= {({\rm Ad}_{H})^{k_{0}}({\rm Ad}_{P_{1}})^{k_{1}}\cdots ({\rm Ad}_{P_{s}})^{k_{s}}} \; .  
\]
With this notation, 
\[
		( D^{k} \varphi_{\rm I})(x)
		=(-1)^{k_{0}}  (-i)^{ | k | -k_{0} }    
		({\rm Ad}_{P})^{k}   (\varphi_{\rm I}(x))   \;.
\]

\begin{definition}
The heat-kernel regularized field $\varphi_{\epsilon}(x)$ and the corresponding heat-kernel regularized 
imaginary-time field $\varphi_{\rm I, \epsilon}(x)$ are
	\[
		\varphi_{\epsilon}(x)=e^{-\epsilon H}\,\varphi(x)\,e^{-\epsilon H}\;
		\text{and} \; 
		\varphi_{\rm I,\epsilon}(x)=e^{-\epsilon H}\,\varphi_{\rm I}(x)\,e^{-\epsilon H} \; ,  \quad \text{$0<\epsilon$} .
	\]
\end{definition}

\begin{theorem}
\label{Analyticity}
Assume the standard hypotheses of \S\ref{Paragraph:SH}.  Then $\varphi_{\epsilon}(x)$   is a bounded operator on 
$\mathcal{H}$ for  real $x$, and it extends to a holomorphic, operator-valued function $\varphi_{\epsilon}(z)$ 
with $z=(z_{0}, \vec z \, )$ in the domain 
\[
	\mathcal{K}_{\epsilon}=\{ z \in \mathbb{C}^d
: |\Im z_{0}| + |\Im\vec z \, | < \tfrac{\epsilon}{4M}\} \; . 
\]
In this domain, $\varphi_{\rm I,\epsilon}(z)= \varphi_{\epsilon}(iz_{0}, \vec z \,)$.  
Also  
\begin{equation}
		\| D^{k}\varphi_{\rm I,\epsilon}(z) \|
		\le M_{1}\left( \tfrac{4M}{\epsilon} \right)^{|k|+\gamma} \, ( | k |+\gamma)! \quad
	\label{Uniform Imaginary-Time Operator Bound}
\end{equation}
for all $z\in \mathcal{K}_{\epsilon}$ and all $k\in\mathbb{N}^{d}_{0}$,
where $M_{1}<\infty$ was introduced in \eqref{M1} and  $\gamma$ is an integer. Both $M_{1}$ and $\gamma$ depend
only on $M,\alpha,d$.

\end{theorem}

\begin{proof}[\bf Proof]
Let $h$ be a $C^{\infty}_{0}$ test function depending only on the $s=d-1$ spatial variables. Then the field is an operator-valued 
distribution for $|t| <\epsilon$ and one can integrate by parts to obtain the identity
\begin{eqnarray*}
		i^{ | k-k_0| }\left( D^{k_0}\varphi_{\rm I,\epsilon}\right) (t,D^{k-k_0}_{\vec x} h)
		&=&  \left( D^{k}\varphi_{\rm I,\epsilon} \right) (t, h)
		\nonumber \\
		&=&  (-1)^{k_{0}}(-i)^{ | k | -k_{0} } \left( {\rm Ad}_{P} \right)^{k} \left( \varphi_{\rm I,\epsilon}(t,h) \right)\;.
\end{eqnarray*}
One can expand the commutators $ \lrp{\rm Ad_{P}}^{k} \lrp{\varphi_{\rm I,\epsilon}(t,h)}$ into  $2^{|k|}$ terms 
$C_{i, \epsilon, t}$.  Each such term is exactly $|k|$ multilinear in the components of $P$, 
namely the Hamiltonian  $H$ and the $s=d-1$ components~$P_{j}$, $j = 1, \ldots, s$,  of the spatial momentum.  
Suppose that in $C_{i, \epsilon, t}$ 
exactly  $\alpha_{i}$ of these factors occur to the left of 
$\varphi_{\rm I,\epsilon}$ and $\beta_{i}=|k| -\alpha_{i} $ occur to the right. 
As $H$ and the components  $P_{j}$ of $P$ mutually commute, each term~$C_{i, \epsilon, t}$ satisfies a  
bound related to a power of $H$.  Use the assumption \eqref{Fundamental Estimates-1} and \eqref{Fundamental Estimates-2}, 
and assume that $1\le M$.  Then, for $0\le t<\epsilon/2$,
\begin{eqnarray*}		
\|  C_{i, \epsilon, t} \| 
		&\le&
		 \sup_{i}\| e^{-\epsilon H/2} (H+I)^{(\alpha_{i}+\tfrac{1}{2})} \|
			\times
			\nonumber
			\\
		&&
			\qquad \times \left\| (H+I)^{-(\alpha_{i}+\tfrac{1}{2})}
			C_{i, 0, 0}   (H+I)^{-(\beta_{i}+\tfrac{1}{2})} \right\|
			\times
			\nonumber
			\\
		&&	 \qquad \qquad \qquad \qquad \times 
			 \| e^{-\epsilon H/2} (H+I)^{(\beta_{i}+\tfrac{1}{2})} \| 
			\nonumber
			\\
		&\le&  
		  \sup_{i} \left(\tfrac{2}{\epsilon}\right)^{\alpha_{i}+\beta_{i}+1}    (\alpha_{i}+ \beta_{i}+1)!\, M^{|k| }\, 
		  			\times
			\nonumber
			\\
			&&	 \qquad \qquad \qquad \qquad \times 
			\left\| (H+I)^{-1/2}\varphi(0,h)(H+I)^{-1/2} \right\| 
			\nonumber
			\\
		&\le&  \left(\tfrac{2M}{\epsilon} \right)^{|k| +1} (|k|+1) ! \;   \|h \|_{\alpha}\;.
\end{eqnarray*}
As  $\mathcal{O}$ is a fixed bounded open set, any Schwartz-space function is in $C^{\infty}_{0}(\mathcal{O})$ and any Schwartz-space  norm satisfies
	\[
		\| h \|_{\alpha}
		\le M_{2}
		\| \lrp{-\Delta_{\vec x} +1}^{r}h \|_{L_{2}(\mathbb{R}^{d-1})}\;,
	\]
for some constants $r$ and $M_{2}$ that depend on the original norm and on~$\mathcal{O}$.   Write 
	\[
		h=\lrp{-\Delta_{\vec x}+1}^{r+[(d+1)/2]}\lrp{-\Delta_{\vec x}+1}^{-r-[(d+1)/2]}h \; , 
	\]
with $[\ \cdot\ ]$ denoting  the integer part.  One can take the derivatives in each term in the multinomial expansion  of 
$\lrp{-\Delta_{\vec x}+1}^ {r+[(d+1)/2]} $  to be the derivatives $D^{k'}$.   Then, using 
	\[
		\| \lrp{-\Delta_{\vec x}+1}^{-r-[(d+1)/2]}h \|_{\alpha}
		\le 
		M_{3}\| \lrp{-\Delta_{\vec x}+1}^{-[(d+1)/2]}h \|_{L^{2}(\mathbb{R}^{d-1})}\;,
	\]
one has, with a  new constant $M_{3}$,  and with $\gamma=2r+d+2$, 
	\begin{multline*}
		\|e^{-\epsilon H}\, ( D^{k}\varphi_{\rm I}) (0, h)   \,e^{-\epsilon H} \|
		\le 2^{|k|} M_{3}\, \lrp{\tfrac{2M}{\epsilon}}^{|k|+\gamma} (|k|+\gamma)! \times \ 
		\\
		\times \| \lrp{-\Delta_{\vec x}+1}^{-[(d+1)/2]}h \|_{L_{2}}\;.
	\end{multline*}
For spatial dimension $s$, the Dirac measure $\delta_{\vec x}$ is an element of the Sobolev space with norm 
	\[
	\| h \|_{H^{-(d+1)/2}} = \| \lrp{-\Delta_{\vec x}+1}^{-[(s+2)/2]}h \|_{L_{2}(\mathbb{R}^{s})}\; . 
	\]  
Thus we can take a limit $k \to \infty$ of $C^{\infty} $-functions $h_k$ supported in~$\mathcal{O}$ and  converging to $\delta_{\vec x}$ for $\vec x\in\mathcal{O}$.  As a consequence, the operator 
$e^{-\epsilon H}\, \lrp{D^{k}\varphi_{\rm I}}(0, h)   \,e^{-\epsilon H}$ extends by continuity to 
$h=\delta_{\vec x}$.  Likewise the same bound holds for the extension, so  with  a new constant $M_{4}$,
	\[
		\left\| e^{-\epsilon H}\, \lrp{D^{k}\varphi_{\rm I}}(0, \vec x)   \,e^{-\epsilon H} \right\|_{\mathcal{H}}
		\le M_{4} \lrp{\tfrac{4M}{\epsilon}}^{|k|+\gamma} (|k|+\gamma)!\;.
	\]
This is the claimed bound \eqref{Uniform Imaginary-Time Operator Bound} on the heat-kernel regularized field with $x\in\mathcal{O}$, with the constant $M_{1}=M_{4}$.  
As the momentum $\vec P$ and $H$ commute, and as time-translation is unitary 
on $\mathcal{H}$, this bound also holds for arbitrary spatial points $\vec x$. 

As a consequence of this estimate in the derivatives of the field, the heat-kernel regularized, imaginary-time field is real-analytic in a neighbourhood of $|t|< \epsilon/ (4M)$ and for all $\vec x\in\mathbb{R}^{d}$.  In fact, the power series for an analytic continuation converges in a neighbourhood of zero, with radius of convergence $\epsilon/ (4M)$. The unitary action of spatial translations, that commute with $H$, leads to analyticity for all real~$\vec x$. 
\end{proof}	

\begin{lemma}
Assume the hypotheses of \Tref{Analyticity}. Let $\widehat A, \widehat B\in\mathcal{H}$, and $0<\epsilon$.  Let  $M, M_{1}$, and $\gamma$ 
be the constants in the theorem cited, and define 
\[
		F(x)
		=\langle \widehat A, e^{-\epsilon H} \varphi_{\rm I}(x)\,e^{-\epsilon H}\widehat B \rangle_{\mathcal{H}}\;.
\]
Then $F(x)$ is real-analytic for $|t|<\epsilon/(4M)$  and all real $\vec x$,  and $F(x)$ extends to a 
holomorphic function in the domain $z=(z_{0},\vec z \, )$ with $|z_{0}|+|\Im\vec z \, |< \epsilon/(4M)$.   
Furthermore, the analytic continuation to this domain satisfies the uniform bound
\[
		|F(z)| \le \tfrac{M_{1}}{\epsilon^{\gamma}} \| \widehat A \|_{\mathcal{H}}\,\|\widehat B \|_{\mathcal{H}}\;.
\]
\end{lemma}

\begin{proof}[\bf Proof]
One can bound $|D^{k}F(x)|\le \| e^{-\epsilon H} D^{k}\varphi_{\rm I}(x)e^{-\epsilon H}\| \|\widehat A\|_{\mathcal{H}} \,
\|\widehat B \|_{\mathcal{H}}$ using the operator norm bounds of \Tref{Analyticity}.  This shows that the power series for $F(x)$ converges absolutely in the desired domain.  The bound \eqref{Uniform Imaginary-Time Operator Bound} then yields  the uniform estimate.  
\end{proof}

\subsection{Quantization Domain}

For $\mathcal{O}\subset X_{+}$,  let $\mathcal{A}(\mathcal{O})$ denote the algebra of polynomials in fields averaged with $C^{\infty}$-functions supported in~$\mathcal{O}$.  

\begin{definition}
We call  $\mathcal{O}\subset\mathbb{R}^{d}_{+}$ a {\em quantization domain}, if the quantization map $A\mapsto \widehat A$ given in 
Equ.~(\ref{equiv-class}) takes the  linear subspace $\mathcal{D}(\mathcal{O})=\mathcal{A}(\mathcal{O}) \Omega_{0}^{\tt E}$ 
into a subspace $\widehat{\mathcal{D}(\mathcal{O})}\subset\mathcal{H}$ that is dense in~$\mathcal{H}$.  
\end{definition}

\begin{theorem}[\bf Non-trivial Quantization Domains]\label{Theorem:Quantization Domain 2}
With the  hypotheses of \S\ref{Paragraph:SH}, any open set $\mathcal{O}\subset X_{+}$  is a quantization domain. 

\end{theorem}

\begin{proof}[\bf Proof] 
If $\mathcal{O}'$ is a quantization domain, then so is any larger set $\mathcal{O}\supset\mathcal{O}'$; therefore  it is no loss of generality to take $\mathcal{O}$ to be a bounded, open set of  small diameter.  
Let $\chi\in\mathcal{H}$ be any vector  that is orthogonal to $\widehat{B}$ for all $B\in \mathcal{D}(\mathcal{O})$, i.e., 
\begin{equation}
		\langle\chi, \widehat {B} \rangle_{\mathcal{H}}=
		0 
		\quad \text{for all} \quad
		B\in \mathcal{D}(\mathcal{O}) \;.
	\label{Vanishing}
\end{equation}
The theorem holds if and only if  such a vector $\chi$ must vanish, so we now show  that $\chi=0$.  
	
It is sufficient to consider $A \in \mathcal{A}(\mathcal{O})$ of the form $A_{n}(f_{1}, \ldots, f_{n})  =\Phi(f_{1})\cdots\Phi(f_{n}) $, 
with each $f_{j}$ supported in the domain~$\mathcal{O}$, and with arbitrary $n$.  Define 
$ A_{n}(x_{1},\ldots,x_{n}) =\Phi({x_{1})}\cdots\Phi({x_{n})}$,  where now the subscript 
labels different $d$-vectors. Set $B_n = A_{n} \Omega_{0}^{\tt E}$ and 
denote the quantization of $B_{n}  $ by  $\widehat{ B_{n}}(x_{1},\ldots,x_{n}) $.    
Note  both $A_{n}$ and $\widehat {B_{n}}$  are symmetric under permutations of the 
coordinates  $\Xn{x}\to x_{\pi_{1}}, \ldots, x_{\pi_{n}}$, for $\pi\in S_{n}$.  Also  
	\[
		 \widehat {A_{n}}(f_{1}, \ldots,f_{n})
		= \int \widehat {A_{n}}(x_{1},\ldots, x_{n}) f_{1}(x_{1})\cdots f_{n}(x_{n})\,dx_{1}\cdots dx_{n}
	\]
satisfies
\begin{equation}
		\langle\chi, \widehat{ B_{n}}\Xnp{f}\rangle_{\mathcal{H}}=0\;, \quad B_n \Xnp{f} = A_{n} \Xnp{f} \Omega_{0}^{\tt E} \; .
	\label{Orthogonality-n}
\end{equation}
 The vector  $\widehat{ B_{n}}(x_{1},\ldots,x_{n})$ is the anti-time-ordered product of the imag\-inary-time fields, 
\begin{equation}
		\widehat{ B_{n}}(x_{1},\ldots,x_{n})
		= \varphi_{\rm I}(x_{i_{1}})\,\cdots\, \varphi_{\rm I}(x_{i_{n}})\,\Omega\;,
		\quad \text{where $ t_{i_{1}}\le \cdots\le t_{i_{n}}$.}		
	\label{Polynomial Field Function}
\end{equation}
As  \eqref{Polynomial Field Function}  is symmetric under permutations of the coordinates,  we need only  consider the case $t_{1}<\cdots<t_{n}$, which we now assume.   Furthermore, we can choose $0<\epsilon$ sufficiently small, so that each $f_{j}$ has support in the set  $\mathcal{O}_{j}\subset\mathcal{O}$ with $\mathcal{O}_{j}$ lying in the time-strip $[T+3j\epsilon, T+3j\epsilon + \epsilon]$, for $j=1,\ldots,n$.  This ensures a time-separation $t_{j+1}-t_{j}\ge2\epsilon$  between $x_{j+1}$ and $x_{j}$ for each $j$.   Hence for $x_{j}\subset \mathcal{O}_{j}$, \Tref{Analyticity} ensures that $\widehat{ B_{n}}(x_{1},\ldots,x_{n}) $ is a real analytic function of $x_{1}, \ldots, x_{n}$ and it extends to a holomorphic function in the complex domain   $|z_{j,0}|+|\Im\vec z_{j}| <\epsilon/ (4M) $ for $j=1, \ldots,n$.   In this domain the extension satisfies the uniform bound 
\[
		\| \widehat B_{n}(z_{1},\ldots,z_{n})\|_{\mathcal{H}}
		\le \left( \tfrac{M_{1} }{\epsilon^{\gamma}}\right)^{n}\;.
\]

For each $j$ fixed take a sequence of $C^{\infty}_{0}$ functions $f_{j, \ell}$ for $\ell=1,\ldots,$ converging to the delta function, i.e., 
$f_{j,\ell} \to \delta_{x_{j}}$ with $x_{j}\in\mathcal{O}_{j}$.  Then  the property \eqref{Orthogonality-n} ensures that the holomorphic function
\[
		F(z_{1},\ldots,z_{n}) = \langle\chi, \widehat{B_{n}} (z_{1}, \ldots, z_{n}) \rangle_{\mathcal{H}}
\]
vanishes in the analyticity domain  $| z_{j,0}| +| \Im\vec z_{j}|<\epsilon(4M)^{-1}$, and hence in any analytic continuation of this domain.  
	
The symmetry of $\widehat{ B_{n}}(x_{1},\ldots,x_{n})$ under $x_{j}\to x_{\pi_{j}}$ for $\pi\in S_{n}$ shows that the  analyticity and the bound on $F_{n}(z_{1},\ldots,z_{n})$ extends to the domain obtained by the permutation $z_{j}\to z_{\pi_{j}}$.   Similarly, analyticity and the uniform bound both extend to the union of these domains over $T>0$.  These domains include all real points $(x_{1},\ldots,x_{n})$ with positive, non-coinciding times.
\end{proof}


\begin{thebibliography}{ABC}
\bibitem{Glimm-Jaffe-Uniqueness} James Glimm and Arthur Jaffe, 
The $ \lambda \Phi^{4}_{2}$ quantum field theory without cutoffs: II. The field operators and the approximate vacuum, 
{\em Annals of Mathematics} {\bf 91}  (1970) 362--401.

\bibitem{Glimm-Jaffe-Spectrum-c} James Glimm and Arthur Jaffe, The $\lambda (\phi^{4})_{2}$ 
quantum field theory without cut-offs. IV.~Perturbations of the Hamiltonian, 
{\em J. Math. Phys.}  {\bf 13} (1972) 1568--1584, see p. 1584.

\bibitem{Glimm-Jaffe}  James Glimm and Arthur Jaffe, {\em Quantum
Physics}, A Functional Point of View, Springer, Berlin-Heidelberg-New York (1981).

\bibitem{Haag}  Rudolf Haag, {\em Local Quantum Physics: Fields, Particles, Algebras}, 
Springer, Berlin-Heidelberg-New York  (1992).

\bibitem{Lectures}  Arthur Jaffe,  {\em Introduction to Quantum Field Theory}, 2005 E.T.H. Lecture Notes.

\bibitem{OS}  Konrad Osterwalder and Robert Schrader, 
Axioms for Euclidean Green's functions, {\em Commun. math. Phys.} {\bf 31} (1973) 83--112.

\bibitem{OS2}  Konrad Osterwalder and Robert Schrader, Axioms for Euclidean Green's functions. II.,
{\em Commun. math. Phys.} {\bf 42} (1975) 281--305.

\bibitem{Reeh-Schlieder}  Helmut Reeh and Siegfried Schlieder, Bemerkungen zur Unit\"ar\"aquivalenz von 
Lorentzinvarianten Feldern, {\em Nuovo Cimento} {\bf 22}, No. 5 (1961) 1051--1056.

\bibitem{Streater-Wightman}  Raymond Streater and Arthur Wightman, {\em PCT, Spin and Statistics, and All That},
Benjamin, New York (1964).

\end{thebibliography}
\end{document}